\documentclass[11pt]{article}
\usepackage{cite}
\usepackage{graphicx}
\usepackage{amsmath}
\usepackage{amsfonts}
\usepackage{amssymb}
\usepackage{color}
\usepackage{latexsym}
\usepackage{algorithm}
\usepackage[noend]{algorithmic}
\usepackage{multirow}
\usepackage{enumerate}
\usepackage{ mathrsfs}
\usepackage{subfigure}
\usepackage[margin=1.3in]{geometry}

\usepackage{amsthm}

\newtheorem{theorem}{Theorem}[section]
\newtheorem{corollary}[theorem]{Corollary}
\newtheorem{lemma}[theorem]{Lemma}
\newtheorem{observation}[theorem]{Observation}
\newtheorem{definition}[theorem]{Definition} 

\newtheorem{conjecture}[theorem]{Conjecture}
\newtheorem{claim}[theorem]{Claim}

\newcommand{\C}{{{\cal{C}}}}
\newcommand{\G}{{{\cal{G}}}}

\newcommand{\old}[1]{{}}

\newcommand{\fullVer}[1]{{}}

\title{Dual Power Assignment via Second Hamiltonian Cycle
}

\author{A. Karim Abu-Affash \ \ \ Paz Carmi \ \ \ Anat Parush Tzur \\
\\
{\small Department of Computer Science, Ben-Gurion University, Israel} \\
{\small {\tt $\{$abuaffas,carmip,parusha$\}$@cs.bgu.ac.il}}\\ }

\begin{document}

\pagestyle{plain}
\maketitle

\begin{abstract}
A \emph{power assignment} is an assignment of transmission power to each of the wireless nodes of a wireless network, so that the induced graph satisfies some desired properties. The \emph{cost} of a power assignment is the sum of the assigned powers. In this paper, we consider the dual power assignment problem, in which each wireless node is assigned a high- or low-power level, so that the induced graph is strongly connected and the cost of the assignment is minimized. 
We improve the best known approximation ratio from $\frac{\pi^2}{6}-\frac{1}{36}+\epsilon\thickapprox 1.617$ to $\frac{11}{7}\thickapprox 1.571$. 

Moreover, we show that the algorithm of Khuller et al.~\cite{Young02} for the strongly connected spanning subgraph problem,
which achieves an approximation ratio of $1.61$, is $1.522$-approximation algorithm for symmetric directed graphs. 
The innovation of this paper is in achieving these results via utilizing interesting properties for the existence of a second Hamiltonian cycle.  
 
\end{abstract}


\section{Introduction}
Given a set $P$ of wireless nodes distributed in a two-dimensional plane, a {\em power assignment} (or a range assignment), in the context of wireless networks, is an assignment of transmission range $r_u$ to each wireless node $u\in P$, so that the induced communication graph has some desired properties, such as strong connectivity. The \emph{cost} of a power assignment is the sum of the assigned powers, i.e., $\sum_{u\in P}r_u^\alpha$, where $\alpha$ is a constant called the {\em distance-power gradient} whose typical value is between $2$ and $5$. A power assignment induces a (directed) \emph{communication graph} $G=(P,E)$, where a directed edge $(u,v)$ belongs to the edge set $E$ if and only if $|uv| \le r_u$, where $|uv|$ is the Euclidean distance between $u$ and $v$. The communication graph $G$ is \emph{strongly connected} if, for any two nodes $u, v \in P$, there exists a directed path from $u$ to $v$ in $G$. In the standard power assignment problem, one has to find a power assignment of $P$ such that (i) its cost is minimized, and (ii) the induced communication graph is strongly connected.

When the available transmission power levels for each wireless node are continuous in a range of reals, many researchers have proposed algorithms for the strong connectivity power assignment problem~\cite{Chen89, Clementi00, Clementi99, Kirousis00, Lloyd02}. In particular, $2$-approximation algorithms based on minimum spanning trees were proposed in~\cite{Chen89, Kirousis00}. When the wireless nodes are deployed in the $2$-dimensional or the $3$-dimensional space, the problem is known to be NP-hard~\cite{Clementi99, Kirousis00}. A survey covering many variations of the problem is given in~\cite{Clementi02}.

In this paper, we study a dual power assignment version, in which each wireless node can transmit in one of two ({\em high} or {\em low}) transmission power levels. Let $r_H$ and $r_L$ denote the transmission ranges of the high- and low-transmission powers, respectively. Since assigning more wireless nodes with the high power level results in a larger power consumption, the objective in the dual power assignment problem is equivalent to minimizing the number of wireless nodes that are assigned high-transmission range $r_H$.

The dual power assignment (DPA) problem was shown to be NP-hard~\cite{Carmi07, Rong04}. Rong et al.~\cite{Rong04} gave a $2$-approximation algorithm, while Carmi and Katz in ~\cite{Carmi07} gave a $9/5$-approximation algorithm and a faster $11/6$-approximation algorithm. Later, Chen et al.~\cite{Chen05} proposed an $O(n^2)$ time algorithm with approximation ratio of $7/4$. Recently, Calinescu~\cite{Calinescu10} improved this approximation ratio to $\approx 1.61$, using in a novel way the algorithm of Khuller et al.~\cite{Khuller02, Young02} for computing a minimum strongly connected subgraph.

A related version asks for a power assignment that induces a connected (also called ``symmetric'' or ``bidirected'') graph. This version is also known to be NP-hard. 
The best known approximation algorithm is based on techniques that were applied to Steiner trees, and achieves approximation ratio of $3/2$~\cite{Nutov09}.

\subsection{Our results}
We present a conjecture regarding an interesting characterization for the existence of a second Hamiltonian cycle and its applications. 
We prove the conjecture for some special cases that are utilized 
(i) to improve the best known approximation ratio for the DPA problem from $\frac{\pi^2}{6}-\frac{1}{36}+\epsilon\thickapprox 1.617$ to $\frac{11}{7}\thickapprox 1.571$, and
(ii) to show that the algorithm of Khuller et al.~\cite{Young02} for the strongly connected spanning subgraph problem,
which achieves a approximation ratio of $1.61$, is $1.522$-approximation algorithm for symmetric unweighted directed graphs. 
Moreover, the correctness of the aforementioned conjecture implies that the approximation algorithm of Khuller et al. is actually a  $3/2$-approximation algorithm in symmetric unweighted digraphs.


\section{Second Hamiltonian Cycle}
A cycle in a graph is Hamiltonian if it visits each node of the graph exactly once; if a graph contains such a cycle, it is called a Hamiltonian graph. Deciding whether a graph is Hamiltonian has been shown to be NP-hard. 
A Hamiltonian graph $G$ contains a second Hamiltonian cycle (\textsc{SecHamCycle} for short) if there exist two Hamiltonian cycles in $G$ that are differed by at least one edge.
A classic result of Smith~\cite{Tutte46} states that each edge in a $3$-regular graph is contained in an even number of Hamiltonian cycles. Thomason~\cite{Thomason78} extended Smith's theorem to all graphs in which all nodes have an odd degree (Thomason's lollipop argument). 
In addition, Thomassen~\cite{Thomassen98} showed that every Hamiltonian $r$-regular graph, where $r \ge 72$, contains \textsc{SecHamCycle}. 
This bound on $r$ was reduced to $23$ by Haxell et al.~\cite{Haxell07}. 

All these related works have considered the existence of \textsc{SecHamCycle} on the whole set of nodes.
In this section, we consider the existence of \textsc{SecHamCycle} also with respect to a subset of the nodes.

Let $G=(V,E)$ be a connected graph and let $\Gamma$ be a subset of $V$.
We say that $G$ contains a Hamiltonian cycle on $\Gamma$ if there exists a simple cycle in $G$ whose nodes are exactly the nodes of $\Gamma$, i.e, the subgraph induced by $\Gamma$ is a Hamiltonian graph.
A cycle in $G$ is $\Gamma$-Hamiltonian with respect to $\Gamma$ if there exists a subset of nodes $U \subseteq (V \setminus \Gamma)$ such that $G$ contains a Hamiltonian cycle on $\Gamma \cup U$. 
We denote such a cycle by $H_G(\Gamma)$;
If $G$ contains $H_G(\Gamma)$, then it is called a $\Gamma$-Hamiltonian graph.
Moreover, we say that $G$ contains a second $\Gamma$-Hamiltonian cycle (\textsc{Sec}-$\Gamma$-\textsc{HamCycle} for short), if $G$ contains a Hamiltonian cycle $H$ on $\Gamma$ and a $\Gamma$-Hamiltonian cycle $H_G(\Gamma)$, that are  differed by at least one edge.

Fleischner~\cite{Fleischner94} constructed a $3$-regular graph $G$ that has a dominating cycle $\Gamma$, such that no other  
\textsc{Sec}-$\Gamma$-\textsc{HamCycle} exists. Below, we conjecture that replacing the regularity requirement with a connectivity requirement, implies the existence of \textsc{Sec}-$\Gamma$-\textsc{HamCycle}.   
  
\begin{conjecture}\label{con:conjectureMain}
Let  $G=(V,E)$ be a connected graph and let $\Gamma \subseteq V$, such that $G$ contains a Hamiltonian cycle $H$ on $\Gamma$ and the graph $(V, E \setminus H)$ is connected. Then $G$ contains a \textsc{Sec}-$\Gamma$-\textsc{HamCycle}.
\end{conjecture}

The following conjecture, which is a special case of Conjecture~\ref{con:conjectureMain}, is shown in Lemma~\ref{lem:lemma2.4} to be actually equivalent, i.e., the correctness of Conjecture~\ref{con:conjecture2.1} yields the correctness of Conjecture~\ref{con:conjectureMain}.

\begin{conjecture}\label{con:conjecture2.1}
Let $H$ be a Hamiltonian cycle on a set of nodes $V$. Every connected bipartite graph $G_b=(V, U, E)$ admits that 
the graph $G=(V \cup U, H \cup E)$ contains a \textsc{Sec}-$V$-\textsc{HamCycle}.
\end{conjecture}

Notice that if two consecutive nodes in $H$ share a common adjacent node of $U$ in $G_b$, then Conjecture~\ref{con:conjecture2.1} is obviously true. Thus, we assume that no such two nodes exist. In addition, since nodes of $U$ of degree $1$ (in $G_b$) can be removed without affecting the correctness of the conjecture, we may assume that each node in $U$ is of degree at least $2$. Finally, we may assume that $G_b$ is a tree. In the following lemmas, we prove Conjecture~\ref{con:conjecture2.1} for some special cases that are essential for proving Theorem~\ref{theo:theorem1approxRatio} in the sequel section.

\begin{lemma} \label{lem:lemma2.1}
If each node in $U$ is of degree $2$, then the conjecture is true.
\end{lemma}
\begin{proof}
Since each $u\in U$ is connected to two nodes of $V$, $G_b$ can be converted to a spanning tree $T=(V,E_T)$ of $V$ by connecting any two adjacent nodes of a node $u \in U$ via an edge and deleting $u$ and the edges incident to it; that is, $G=(V, H \cup E_T)$. We distinguish two cases:
\begin{itemize}
	\item {\bf $|V|$ is even:} decompose $T$ into a forest $T'=(V,E_{T'})$ s.t. each node of $V$ has an odd degree in $T'$. The existence of such a decomposition can be easily proven by induction on $|V|$. 
The graph $G'=(V, H \cup E_{T'})$ is a Hamiltonian graph with nodes of odd degree; 
therefore, by Thomason's lollipop argument~\cite{Thomason78}, it contains a \textsc{SecHamCycle} on $V$ that yields a \textsc{SecHamCycle} on $V$ in $G$. 
\item {\bf $|V|$ is odd:} duplicate $G$ to get a new graph $G_{d}=(V\cup V',H\cup E' \cup E_T\cup E'_T)$, in which $V'$ is a copy of $V$, and, for each edge $\{v_i,v_j\}\in H$ (resp., $\{v_i,v_j\}\in E_{T}$), there is an edge $\{v'_i,v'_j\} \in E'$ (resp., $\{v'_i,v'_j\} \in E'_T$). Let $v_i$ and $v_j$ be two consecutive nodes in $H$. Connect $v_i$ (resp., $v_j$) to its duplicated node $v'_i$ (resp., $v'_j$) by an edge denoted by $e_i$ (resp., $e_j$), and connect $v_j$ to $v'_i$ by an edge. Finally, remove from $G_d$ the edges $\{v_i,v_j\}$ and $\{v'_i,v'_j\}$. The obtained graph $G_d$ contains a Hamiltonian cycle and a spanning tree on $V\cup V'$ that are edge disjoint. By case 1, since $|V\cup V'|$ is even, we conclude that $G_d$ contains a \textsc{SecHamCycle} on $V\cup V'$; that contains $e_i$ and $e_j$, and yields a \textsc{SecHamCycle} on $V$ in $G$.
\end{itemize}
\end{proof}

\begin{claim}~\label{claim:claim2.1} 
Let $T=(V, U,E_T)$ be a bipartite spanning tree of $V\cup U$, s.t. $|V|$ is even and all nodes of $U$ are of degree $2$ or $3$. Then, there exists a forest $T'=(V, U,E^{'}_{T})$, in which (i) $E^{'}_{T} \subseteq E_T$, (ii) each node in $V$ is of odd degree, and (iii) each node in $U$ is of degree $2$. 
\end{claim}

\begin{proof}
The claim can be proven by an induction on the number of nodes of degree $3$ in $U$. Consider a node $u \in U$ of degree $3$ that is connected to three nodes $v_i,v_j$ and $v_k$ from $V$. Since $|V|$ is even, at least one of the three subtrees rooted at $v_i,v_j$ and $v_k$ (and not containing $u$) has an even number of nodes from $V$. Assume w.l.o.g. that the subtree rooted at $v_i$ has an even number of nodes from $V$. Thus, removing the edge $\{u,v_i\}$ from $T$ decomposes $T$ into two subtrees each has less number of nodes of degree $3$ from $U$ than $T$. Once we have a forest of subtrees each has even number of nodes from $V$ and each node from $U$ has a degree $2$, we can convert it to a forest $T'$ as in case 1 in the proof of Lemma~\ref{lem:lemma2.1}. 
\end{proof}

By this claim and by Lemma~\ref{lem:lemma2.1}, we have the following lemma.
\begin{lemma}\label{lem:lemma2.2}
If each node in $U$ is of degree at most $3$, then the conjecture is true.
\end{lemma}

The following corollary obtained by applying the duplication technique from the proof of Lemma~\ref{lem:lemma2.1}.
\begin{corollary}\label{cor:corollary2.1} 
If each node in $U$ is of degree at most $3$, then, for any edge $e$ of $H$, there exists a \textsc{Sec}-$V$-\textsc{HamCycle} in $G$ that contains $e$.
\end{corollary}
\begin{corollary}\label{cor:corollaryForest1} 
Let $H$ be a Hamiltonian cycle on a set of nodes $V$, and let $F=(V, U,E)$ be a bipartite forest, such that 
(i) each node in $U$ is of degree at most 3, 
(ii) each tree in $F$ contains an even number of nodes of $V$.
Then, the graph $G=(U \cup V, H \cup E)$ contains  a \textsc{Sec}-$V$-\textsc{HamCycle}.
\end{corollary}
\begin{corollary}\label{cor:corollaryForest2} 
Let $H$ be a Hamiltonian cycle on a set of nodes $V$, and let $F=(V, U,E)$ be a bipartite forest, such that 
(i) each node in $U$ is of degree at most 3, 
(ii) each tree in $F$ contains an even number of nodes of $V$ \textbf{except of exactly one tree}. 
Then, the graph $G=(U \cup V, H \cup E)$ contains  a \textsc{Sec}-$V$-\textsc{HamCycle}.
\end{corollary}
\begin{proof}
Let $T_{odd} \in F$ be the tree that contains an odd number of nodes of $V$, and let $(v_i, v_j)$ be an edge  of $H$, such that  $v_i \in T_{odd}$.
Consider the duplication technique from the proof of Lemma~\ref{lem:lemma2.1}. Instead of connecting $v_i$ (resp., $v_j$) to its duplicated node $v'_i$ (resp., $v'_j$), we connect $v_i$ (resp., $v_j$) to $v'_j$ (resp., $v'_i$) and $v_i$ to $v'_i$.
Then, by Corollary~\ref{cor:corollaryForest1} we are done.
\end{proof}

%
%
Given a bipartite graph $(V,U,E)$, for a node $v \in V$ and a subset $W \subseteq U$, denote by $N_{v}(W)$ the set of neighbors of $v$ in $W$, 
i.e., $N_{v}(W) = \{u \in W  : \ \{u,v\} \in E \}$. 

\begin{lemma}\label{lem:lemma2.3}
Let $U'$ be the subset of $U$ containing all nodes of degree at least $4$. 
If there exist two consecutive nodes $v_i$, $v_j$ of $H$ such that $N_{v_i}(U') \cup N_{v_j}(U') = U'$, 
then the conjecture is true. 
\end{lemma}

\begin{proof}
Consider the graph $G'_b$ that is obtained from $G_b$ by the following modification. 
Recall that $N_{v_i}(U')  \cap N_{v_j}(U') =\emptyset$. For each node $u' \in N_{v_i}(U')$ 
(resp., $u' \in N_{v_j}(U')$), and for each $v \in V\setminus \{v_i\}$ (resp., $v \in V\setminus \{v_j\}$) that is adjacent to $u'$, we add a new node $u_v$ to $U$ and update the set $E$ to be $E \setminus \{\{u',v\}\} \cup \{\{v_i,u_v\},\{u_v,v\}\}$ (resp., $E\setminus \{\{u',v\}\} \cup \{\{v_j,u_v\},\{u_v,v\}\}$). Then, we remove the edges $\{v_i,u'\}$ (resp., $\{v_j,u'\}$) from $E$, and the node $u'$ from $U$. The obtained graph $G'_b$ is a connected bipartite graph and each node in $U$ is of degree at most $3$; therefore, by Corollary~\ref{cor:corollary2.1}, the graph obtained by adding the edge set $H$ to $G'_b$ contains a \textsc{Sec}-$V$-\textsc{HamCycle} that contains the edge $\{v_i,v_j\}$. 
Thus, $G$ contains \textsc{Sec}-$V$-\textsc{HamCycle}.
\end{proof}
\begin{corollary}\label{cor:cor2.10}
Let $v_i$ and $v_j$ be two nodes of $H$ such that  $N_{v_i}(U') \cup N_{v_j}(U') = U'$, 
If by removing the nodes on one of the two paths between $v_i$ and $v_j$ on $H$ (and their incident edges) from $G_b$, the graph $G_b$ remains connected, then the conjecture is true.
\end{corollary}

\old{
\begin{observation}\label{obs:configur}
Let $e_1,e_2,\ldots,e_n$ be the edges of the Hamiltonian cycle $H$, where the endvertices of the edge $e_i$ are $v_i$ and $v_{i+1}$, $1 \le i \le n-1$, and the endvertices of the edge $e_n$ are $v_n$ and $v_1$. 
If one of the following three configuration exists then $G$ contains a \textsc{Sec}-$V$-\textsc{HamCycle}.
\begin{itemize}
	\item there exist two nodes $u_1,u_2\in U$ and two indices $i<j$, such that  
	$\{v_i,u_1\}$, $\{v_j,u_1\}$, $\{v_{i+1},u_2\}$, $\{v_{j+1}, u_2 \}$ $\in E$ 
	(see Figure~\ref{fig:stru4SecH}, left).
	\item there exist three nodes $u_1,u_2,u_3 \in U$ and three indices $i<j<k$, such that  
	$\{v_i,u_1\}$, $\{v_j,u_1\}$, $\{v_{i+1},u_2\}$, $\{v_{k},u_2\}$, $\{v_{j+1},u_3\}$, $\{v_{k+1},u_3\}$ $\in E$, 
	(see Figure~\ref{fig:stru4SecH}, middle).
	\item there exist three nodes $u_1,u_2,u_3 \in U$ and three indices $i<j<k$, such that 
	$\{v_i,u_1\}$, $\{v_{j+1},u_1\}$, $\{v_{i+1},u_2\}$, $\{v_{k},u_2\}$, $\{v_{j},u_3\}$, $\{v_{k+1},u_3\}$ $\in E$ 
	(see Figure~\ref{fig:stru4SecH}, right).
\end{itemize}
\end{observation}
\begin{figure}[h]
   \centering
       \includegraphics[width=0.95\textwidth]{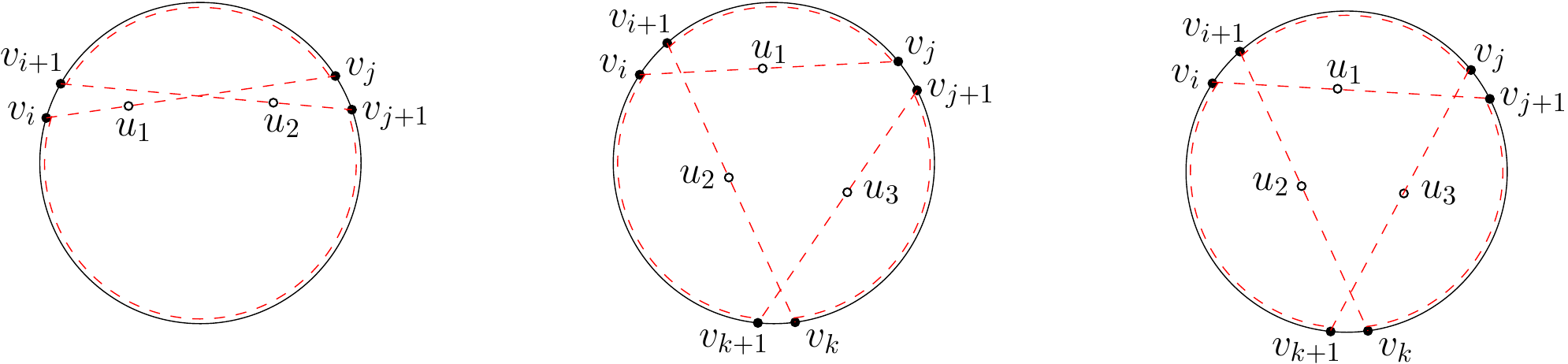}
   \caption{An illustration of configurations that admit a second Hamiltonian cycle.}
   \label{fig:stru4SecH}
\end{figure}
\noindent
\emph{Remark:} Clearly there are many configurations that admit \textsc{Sec}-$V$-\textsc{HamCycle}, however, 
for proving Lemma~\ref{lem:lemma2.5}, it is sufficient to concentrate only on the aforementioned three configurations.
}

\begin{claim}\label{cl:AcycleThatAdimtCycle}
Let $v_i \in V$ be a node such that $|N_{v_i}(U)|= 1$ in $G=(V \cup U, E \cup H)$ (i.e., $v_i$ is a leaf in the tree $(V \cup U, E)$), and  
let $v_{i+1}$ and $v_{i-1}$ be its two neighbors in $H$ (i.e., $\{v_i, v_{i+1} \} , \{v_{i-1}, v_i \} \in H$).
Let $G^*=(V^* \cup U^*, E^* \cup H^*)$ be a graph obtained from $G$ by the following modifications. Assume $N_{v_i} = \{ u \}$, see Figure~\ref{fig:transformation}. 
\begin{eqnarray*}
  V^*  \leftarrow V  & \cup & \{ v_l,v_r \} \\
	U^* \leftarrow U   & \cup  & \{ u' \}  \cup \{ u_j :  \forall v_j \in ( N_u(V)\setminus \{v_i\} ) \}   \setminus \{ u \} \\
	E^* \leftarrow E   & \cup & \{ \{v_l, u' \}, \{v_r, u' \}  \} \\
                     & \cup & \{ \{v_i, u_j\}, \{u_j, v_j \} :  \forall v_j \in N_u(V)   \}  \\
	                   & \setminus & \{ \{u, v_j\} : \forall v_j \in N_u(V) \}  \\
  H^* \leftarrow H & \cup & \{ \{v_{i-1}, v_l \}, \{v_l, v_i \}, \{v_i, v_r \},  \{v_r, v_{i+1} \} \} \\  
                    & \setminus & \{ \{v_{i-1}, v_i \}, \{v_i, v_{i+1} \} \}
\end{eqnarray*}
Then,  \textsc{Sec}-$V^*$-\textsc{HamCycle} in $G^*$ admits a \textsc{Sec}-$V$-\textsc{HamCycle} in $G$.
\end{claim}

\begin{figure}[h]
   \centering
       \includegraphics[width=0.75\textwidth]{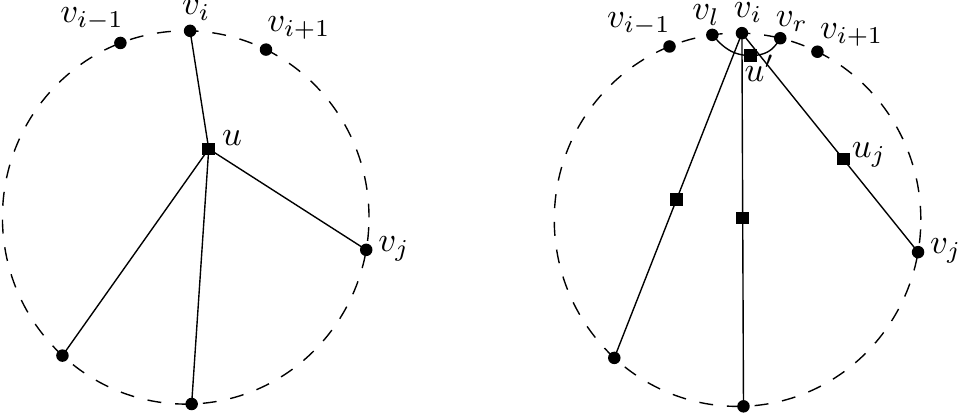}
   \caption{An illustration of the modified graph $G^*$ (on the right) from graph $G$ (on the left), where the edges of $H$ and $H^*$ are dashed, and the edges of $E$ and $E^*$ are solid.}
   \label{fig:transformation}
\end{figure}
\begin{proof}
Let $C^*$ be  \textsc{Sec}-$V^*$-\textsc{HamCycle} in $G^*$, then if $\{v_l, v_i \}, \{v_i, v_r \} \in C^*$, then $C^*$ admits a  \textsc{Sec}-$V$-\textsc{HamCycle} in $G$. Therefore, assume w.l.o.g., that  $\{v_l, v_i \} \notin C^*$;  thus, 
$\{v_l, u' \}, \{u', v_r \} \in C^*$ . We distinguish between two cases:
\begin{itemize}
	\item \textbf{ $\{v_i, v_r \} \in C^*$:}
         The path $P^* = (v_{i-1},v_l, v_r , v_i , u_j)$ is a path in $C^*$. 
         Thus, by replacing the path $P^*$ in $C^*$ with the path $(v_{i-1},  v_i , u)$ in $G$, we have \textsc{Sec}-$V$-\textsc{HamCycle} in $G$.
  \item \textbf{$\{v_i, v_r \} \notin C^*$:}
          The cycle $C^*$ contains two paths $P_1= (v_{i-1},v_l, v_r , v_{i+1} )$ and $P_2=(v_j, u_j, v_i, u'_j, v'_j)$.
          Thus, by replacing the paths $P_1$ and $P_2$ in $C^*$ with the paths $(v_{i-1},v_i, v_{i+1} )$ and 
              $(v_j, u,  v'_j)$ in $G$, respectively, we have \textsc{Sec}-$V$-\textsc{HamCycle} in $G$.
\end{itemize}

\end{proof}

In the next two lemmas we show that the conjecture holds for bounded values of $|V|$.
First, we present a simple proof showing that the conjecture holds for $|V| \le 15$, 
then, we provide a different proof that extends the bound to 23.

\begin{lemma} \label{lem:lemma2.5}
If $|V|\le 15$, then the conjecture is true.
\end{lemma}
\begin{proof}
Let $U'$ be the set of nodes in $U$ of degree at least $4$. Recall that no two consecutive nodes $v_i, v_{i+1}$ in $H$ share a common adjacent node of $U$ in $G_b$ (i.e., $N_{v_i}(U)  \cap N_{v_{i+1}}(U) =\emptyset$). 

If there exists a node $v_i \in V$ such that  $N_{v_i}(U') = U'$, then any adjacent node of $v_i$ in $H$, w.l.o.g. $v_{i+1}$, satisfies  $N_{v_i}(U') \cup N_{v_{i+1}}(U') = U'$, and, by Lemma~\ref{lem:lemma2.3}, we are done.   
Thus, we may assume that no such a node exists, and hence, $|U'| > 1$.  

Recall that we assume that $(V \cup U, E)$ is a tree, thus $|E|= |V|+|U| - 1 = 14 + |U|$.
Moreover, $|E| \geq 4|U'| + 2|U \setminus U'| = 2|U'| + 2|U|$. Hence,   
$2|U'| + 2|U|  \le   14 + |U|$, and we have
\begin{align}\label{formula:numberOfEdgesInU}
2|U'| + |U|  \le  14.
\end{align}
This yields that $|U'| < 5$

We distinguish between the remaining 3 cases of $U'$ cardinality.
\begin{itemize}
	\item  \textbf{$|U'|=2$:} Since $|V| \leq 15$, by the pigeonhole principle, there are two consecutive nodes $v_i,v_{i+1} \in V$ such that $N_{v_i}(U') \cup N_{v_{i+1}}(U') = U'$, and, by Lemma~\ref{lem:lemma2.3}, we are done. 
\item \textbf{$|U'|=3$:} By~(\ref{formula:numberOfEdgesInU}), we have $|U \setminus U'| \le 5$. 
Moreover, the tree $(V \cup U, E)$ has at least 8 leaves. 
Thus, by the pigeonhole principle, there exists a node $v \in V$ such that  $|N_{v}(U)| = |N_{v}(U')| = 1$ (i.e., $v$ is a leaf in the tree $(V \cup U, E)$), and  two consecutive nodes $v_i,v_{i+1} \in V\setminus \{v \}$, such that $N_{v_i}(U') \cup N_{v_{i+1}}(U') = U' \setminus N_{v}(U')$.
Then, by Claim~\ref{cl:AcycleThatAdimtCycle} and by Lemma~\ref{lem:lemma2.3}, we are done. 
\item \textbf{$|U'|=4$:} By~(\ref{formula:numberOfEdgesInU}), we have $|U \setminus U'| \le 2$. 
Moreover, the tree $(V \cup U, E)$ has at least 10 leaves. 
Thus, by the pigeonhole principle, there exist two nodes $v,v' \in V$ such that  $|N_{v}(U)| = |N_{v}(U')| = 1$, 
$|N_{v'}(U)| = |N_{v'}(U')| = 1$ and $N_{v}(U') \neq N_{v'}(U')$, and  two consecutive nodes $v_i,v_{i+1} \in V\setminus \{v,v' \}$, such that $N_{v_i}(U') \cup N_{v_{i+1}}(U') = U' \setminus (N_{v}(U') \cup N_{v'}(U'))$.
Then, by Claim~\ref{cl:AcycleThatAdimtCycle} and by Lemma~\ref{lem:lemma2.3}, we are done. 
\end{itemize}
\end{proof}

In the following lemma we prove that the conjecture holds for $|V| < 24$.  
Actually, we show a stronger claim, that is, we claim that the conjecture holds also for wider family of graphs denoted $\G$. 
Let $\G$ be the family of all graphs $(V \cup U, H \cup E)$, such that 
$(V, U, E)$ is a bipartite graph, where $(V \cup U, E)$ is a forest and 
\begin{enumerate}
	\item[(i)] each tree in $(V \cup U, E)$ has an even number of nodes of $V$, 
\item[(ii)] $H$ is a Hamiltonian cycle on the set of nodes $V$, and 
\item[(iii)] $|V| < 24$.
\end{enumerate}

Notice that, if each graph in $\G$ contains a second Hamiltonian cycle, then this implies that the conjecture is true for the original family of graphs 
(where $(V \cup U, E)$ is a tree) having $|V| < 24$.
\begin{lemma} \label{lem:lemma24}
The conjecture holds for each $G  \in \G$. 
\end{lemma}
\begin{proof}
We prove the lemma by considering a minimal graph in $\G$ that violates the conditions in the above lemmas, claims, and corollaries. 
More precisely, assume that there is a graph in $\G$ that does not contain a second Hamiltonian cycle, and let $G=(V \cup U, H \cup E)$ be a graph in $\G$ that does not contain a second  Hamiltonian cycle, such that the number of nodes in $U$ of degree at least 3 is minimal. 
Let $U' \subseteq U$ be the set of nodes of degree at least 4.
Recall that each node of $U$ is of degree at least $2$.
By the proof of Claim~\ref{claim:claim2.1}, the set $U$ does contain a node of an odd degree, where the proof shows how to reduce the number of nodes of an odd degree (if exits), 
which contradicts the minimality of the number of nodes in $U$ of degree at least 3.

By Lemma~\ref{lem:lemma2.2}, if $|U'|= 0$, then $G$ contains a second Hamiltonian cycle, in contradiction, and, by Lemma~\ref{lem:lemma2.3}, there are no two consecutive nodes $v_i,v_{i+1}$ in $H$ such that $N_{v_i}(U') \cup N_{v_{i+1}}(U') = U'$. Therefore, $U'=\{u_1, \dots , u_k \}$, where $k \geq 2$. 
Moreover, by Claim~\ref{cl:AcycleThatAdimtCycle}
, for each $v \in N_{u_i}(V)$, we have $|N_v(U)| > 1$ (i.e., $v$ is not a leaf in $(V \cup U, E)$), where $u_i \in U'$.
Furthermore, if $|U'| = 2$ (i.e., $U'=\{u_1, u_2 \}$) and $u_1$ and $u_2$ do not belong to the same tree in $(V \cup U, E)$, then, clearly, $|V| \geq 24$, see Figure~\ref{fig:24-example} for illustration. Otherwise, let $v \in N_{u_1}(V)$ and $v' \cup N_{u_2}(V)$ be two nodes, such that $v$ and $v'$ are consecutive nodes in $H$, 
or one of the two paths between $v$ and $v'$ in $H$ consists only of nodes that are leaves in $(V \cup U, E)$. Notice that there are at least two such pairs $v$ and $v'$. 
By Corollary~\ref{cor:cor2.10}, $G$ contains a second Hamiltonian cycle, in contradiction. 
Thus, $V$ must contain at least one additional node for such a pair. 
Therefore, we have that $|V| \geq 24$.
%
%
\begin{figure}[h]
   \centering
       \includegraphics[width=0.85\textwidth]{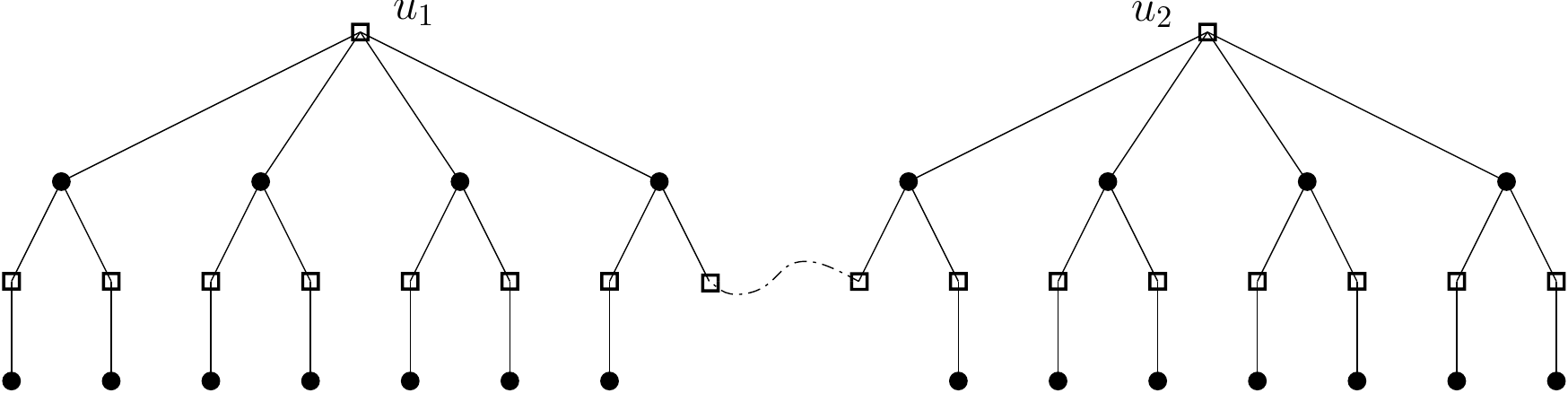}
   \caption{A minimal graph $(V \cup U,E)$ (with respect to $|U'|$) that does not admit a second Hamiltonian cycle 
   by Lemma~\ref{lem:lemma2.2}, Lemma~\ref{lem:lemma2.3}, and Claim~\ref{cl:AcycleThatAdimtCycle}. 
   The circles denote the nodes of $V$ and the squares denote the nodes of $U$. 
   The set $U$ contains at least two nodes ($u_1, u_2$) of degree at least 4, each connected to non-leaf nodes of $V$.}
   \label{fig:24-example}
\end{figure}

Notice that, by extending the aforementioned to the case where $|U'| \geq 3$, we get that $|V|$ is at least 30 (i.e., the minimal graph that follows the above (where $|U'| \geq 3$) has a tree of at least 29 nodes of $V$, however since it needs to be of even number of nodes of $V$, we conclude that $|V| \geq 30$). Thus, we assume that $|U'| = 2$ (i.e., $U'=\{u_1, u_2 \}$).

\end{proof}


In order to apply these lemmas for proving Theorem~\ref{theo:theorem1approxRatio} it is sufficient to prove the following auxiliary lemma.

\begin{lemma} \label{lem:lemma2.4}  
Let $G_{sb}=(V\cup U,E_{sb})$ be a connected graph such that $V$ is an independent set in $G_{sb}$, 
and let $H=(V,E)$ be a Hamiltonian cycle on $V$. 
Then, the graph $G_{sb}$ can be converted to a connected bipartite graph $G_b=(V, U^*,E_b)$ such that $U^* \subseteq U$ and, if the graph $G^*=(V \cup U^*,E \cup E_b)$ contains a  \textsc{Sec}-$V$-\textsc{HamCycle}, then $G=(V\cup U,E\cup E_{sb})$ also contains a \textsc{Sec}-$V$-\textsc{HamCycle}. 
\end{lemma}
\begin{proof}
Let $G_U=(U,E_U)$ be the subgraph of $G_{sb}$ that is induced by $U$, and let $n$ be the number of edges in $E_U$. The proof is by induction on $n$.
\\{\bf Basis:} $n=0$, the claim clearly holds ($G_b=G_{sb}$).
\\{\bf Inductive step:} Let $\{u_i,u_j\} \in E_U$, such that $u_i$ is connected to at least one node $v \in V$.
There exists such a node $u_i$, since the graph $G_{sb}$ is connected. 
Consider the graph $G^*_{sb}=(V, U^*,E^*_{sb})$ that is obtained from $G_{sb}$ by connecting the adjacent nodes of $u_i$ to $u_j$, and removing $u_i$ and the edges incident to it, that is,
\begin{eqnarray*}
    U^*      \  =  \ U  & \setminus &  \{u_i \}  \text{\qquad  and} \\ 
    E^*_{sb}   =  \ E_{sb}  & \cup  & \{ \{u_j,w \} : \forall w \in N_{u_i}(U \cup V) \} \\
                       & \setminus & \{ \{u_i,w \} : \forall w \in N_{u_i}(U \cup V) \}. 
\end{eqnarray*}

By the induction hypothesis, $G^*_{sb}$ can be converted to a connected bipartite graph $G_{b}=(V, U^*,E_b)$ satisfying the lemma.
Thus, since any \textsc{Sec}-$V$-\textsc{HamCycle} $C^*$ in the graph $(V\cup U^*,E\cup E^*_{sb})$ contains at most two edges that are incident to $u_j$ and were generated during the modification of $G_{sb}$, the cycle $C^*$ admits a \textsc{Sec}-$V$-\textsc{HamCycle} in $G=(V\cup U,E\cup E_{sb})$. 
\end{proof}

\section{Dual Power Assignment}\label{sec:dpa}

Let $P$ be a set of wireless nodes in the plane and let $G_R=(P,E_R)$ be the communication graph that is induced by assigning a high transmission range $r_H$ to the nodes in a given subset $R\subseteq P$ and assigning low transmission range $r_L$ to the nodes in $P\setminus R$, and with edge set $E_R=\{(u,v):|uv|\le r_u\}$.

\begin{definition}
A {\bf {\em strongly connected component}} $C$ of $G_R$ is a maximal subset of $P$, such that for each pair of wireless nodes $u,v$ in $C$, there exists a path from $u$ to $v$ in $G_R$.
\end{definition}

\begin{definition}\label{def:definition1}
The {\bf {\em components graph}} $CG_R$ of $G_R$ is an undirected graph in which there is a node $C_i$ for each strongly connected component $C_i$ of $G_R$ (throughout this paper, for convenience of presentation, we will refer to the nodes of $CG_R$ as components, and to the wireless nodes of $G_R$ as nodes). In addition, there exists an edge between two components $C_i$ and $C_j$ if and only if there exist two nodes $u \in C_i$ and $v \in C_j$ such that $|uv| \le r_H$.
\end{definition}

\begin{definition}\label{def:definition2}
A set $Q\subseteq P$ is a {\bf {\em k-contracted set}} of a set $\C=\{C_1,C_2,\ldots,C_k\}$ of $k$ distinct components in $CG_R$ if $|Q\cap C_i| = 1$ for each $C_i \in \C$, and the components in $\C$ are contained in the same strongly connected component in $G_{R\cup Q}$; see Figure~\ref{fig:figure1.1} for illustration.
\end{definition}
\begin{figure}[htb]
   \centering
       \includegraphics[width=0.7\textwidth]{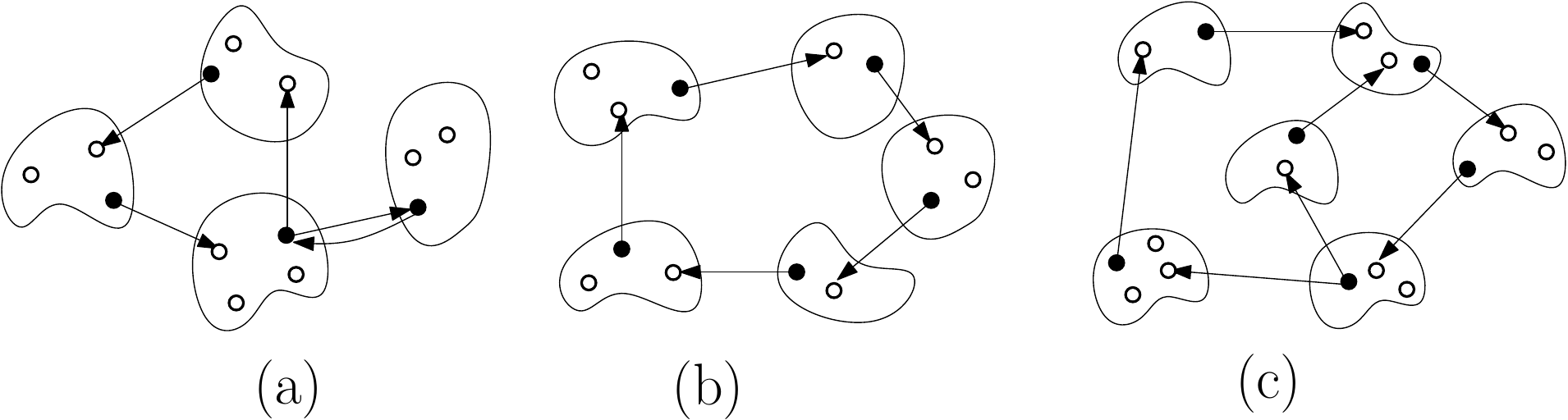}
   \caption{Examples of $k$-contractible structures: (a) $4$-contractible structure, (b) $5$-contractible structure, and (c) $6$-contractible structure. The solid circles in each $k$-contractible structure represent the nodes of the $k$-contracted set of the components.}
   \label{fig:figure1.1}
\end{figure}

Let $\C=\{C_1,C_2,\ldots,C_k\}$ be a set of components in $CG_R$, let $Q$ be a $k$-contracted set of $\C$, and let $v_i$ be the node in $Q\cap C_i$, for each $C_i \in \C$.

\begin{definition}\label{def:definition3}
A {\bf {\em k-contractible structure}} induced by $\C$ and $Q$ is a graph over $\C$ in which there exists a directed edge from $C_i$ to $C_j$ if $v_i$ can reach a node in $C_j$; see Figure~\ref{fig:figure1.1}  for illustration.
\end{definition}
 
\begin{definition}
A {\bf {\em leaf}} in a $k$-contractible structure induced by $\C$ and $Q$ is a component $C_i\in \C$ such that (i) $\C \setminus \{C_i\}$ and $Q\setminus \{v_i\}$ induce a $(k-1)$-contractible structure, (ii) each component in $\C$ is reachable from $C_i$ only via a path containing components from $\C$, and (iii) for each node $u$ in $C_i$, by assigning a high transmission range to $u$, if $u$ reaches a component from $\C$ then every component $C\notin\C$ is not reachable from $u$.
\end{definition}

Given a set $P$ of $n$ wireless nodes in the plane and two transmission ranges $r_L$ and $r_H$ such that the communication graph $G_P$ that is induced by assigning a high transmission range $r_H$ to the nodes in $P$ is strongly connected, in the \emph{dual power assignment} problem the objective is to find a minimum set $R^*\subseteq P$ such that the induced communication graph $G_{R^*}$ is strongly connected. Let $OPT$ denote the size of $R^*$. We present an approximation algorithm that computes a set $R\subseteq P$, such that the graph $G_R$ is strongly connected and the size of $R$ is at most $\frac{11}{7}\cdot OPT$. 

\subsection{Approximation algorithm}
Our algorithm is composed of an initialization and three phases and is based on the idea of Carmi and Katz~\cite{Carmi07} and  Calinescu~\cite{Calinescu10}. 
The main innovation of this algorithm is in achieving a better approximation ratio by utilizing the existence of a second Hamiltonian cycle.
During the execution of the algorithm, we incrementally add nodes to the set $R$ and update the graph $G_R$ accordingly.
The algorithm works as follows.
\\{\bf Initialization.} Set $R=\emptyset$ and compute the induced communication graph $G_R$, i.e., $G_{\emptyset}=(P,E)$, by assigning $r_L$ to each node in $P$ and setting $E=\{(v,u):|vu|\le r_L\}$. 
\\{\bf Phase 1.} While $G_R$ contains a $j$-contracted set, for $j\ge k$ (where $k$ is a constant to be specified later), find a $j$-contracted set, add its $j$ nodes to $R$, and update $G_R$ accordingly.
\\{\bf Phase 2.} Intuitively, we look for contractible structures, where we give priority to those with leaves and then according to their size. 
More precisely, for each iteration $i=k-1, k-2,\ldots,5,4$, while $G_R$ contains an $i$-contracted set, find a contracted set in the following priority order (where $1$ is the highest priority), add its nodes to $R$, and update $G_R$ accordingly (notice that, in each iteration $i$, any contractible structure in $G_R$ is of size at most $i$).
\begin{enumerate}
	\item A $j$-contracted set that induces a contractible structure with at least two leaves, where $j \geq 4 $.
	
	\item A $j$-contracted set that induces a contractible structure with one leaf, such that, if $i > \left\lceil {k}/{2}\right\rceil$ then $j \geq \left\lceil {k}/{2}\right\rceil$, otherwise $j = i$.
	
	\item An $i$-contracted set that induces a contractible structure forming a simple cycle.
	\item An $i$-contracted set that induces a contractible structure of combined cycles.
\end{enumerate}
{\bf Phase 3.} Find a minimum set $R_3 ^* \subseteq P$ such that $G_{R\cup R_3 ^*}$ is strongly connected, and update $R$ to be $R\cup R_3 ^*$. Notice that at the beginning of this phase, any contracted set in $G_{R}$ is of size at most $3$. In Section~\ref{sec:OptG3} we show how to find an optimal solution $R_3  ^*$ for such graphs in polynomial time.

The output of the algorithm is the set $R$, where the resulting graph $G_R$ is strongly connected. In the following section, we analyze the performance guarantee of our algorithm.

\subsection{Time complexity}
An $i$-contracted set can be found naively in $O(n^{i+2})$ time by considering all combinations of sets of nodes of size $i$. 
Moreover, given a constant $k$ finding a contracted set of size greater than $k$ can be found in $O(n^{k+2})$ time. 
For example, a contracted set of size greater than $k$ that induces a simple cycle can be found by considering all paths of 
length $k$ then by checking whether there is a simple path between the path's end-points that avoids the inner nodes of the path.
Finally, since each contracted set reduces the number of components by at least two, the number of contracted sets found by algorithm is $O(n)$.
Thus, the running time of the algorithm is polynomial. 
Notice that for a constant $k$, a $k$-contracted set can be found efficiently
using ideas from Alon et al.~\cite{AlonTZ1995}, where they show how to find simple paths and cycles of a specified length $k$, using the method of \emph{color-coding}.

\subsection{Approximation ratio}

In this section, we prove that the size of $R$ (denoted by $|R|$) at the end of the algorithm is at most $\frac{11}{7}\cdot OPT$. Let $R_i$ denote the set $R$ at the beginning of the $k-i$ iteration of phase 2, for $4 \le i \le k-1$, and let $R_3$ denote the set $R$ at the beginning of phase 3. Given a set ${R_i}$, let $n_i$ denote the number of components of $CG_{R_i}$, and let $OPT(G_{R_i})$ denote the size of a minimum set of nodes $R^*_i \subseteq P$ for which $G_{R_i\cup R^*_i}$ is strongly connected (i.e., $R^*_i$ is an optimal solution for $G_{R_i}$). 
Let $b_i$ (resp., $b_{i,j}$) denote the number of $i$-contracted sets (resp., $j$-contracted sets) found by the algorithm in the $k-i$ iteration.
The following lemma Immediately holds by Definition~\ref{def:definition2}.

\begin{lemma}\label{lem:lemma1.1}
For each $4 \le i < k$, we have \ $$n_i = n_{i-1} + (i-1)\cdot b_i + \sum\limits_{j=4}^{i-1}{(j-1)\cdot b_{i,j}}.$$
\end{lemma}
%
%
\begin{lemma}\label{lem:lemma1.2}
For each $3 \le i < k$, we have \ $$\frac{i}{i-1}(n_i - 1)\le OPT(G_{R_i}) \le 2(n_i - 1).$$
\end{lemma}
\begin {proof}
Let $T$ be a spanning tree of $CG_{R_i}$. For each $\{C_i,C_j\}\in T$, select two nodes $v_i \in C_i$ and $v_j\in C_j$ such that $|v_iv_j|\le r_H$, and add them to $R_i$. Clearly, the resulting communication graph is strongly connected and the cost of this solution is at most $2(n_i - 1)$, which proves the upper bound. The amortized cost of each contracted component of an $i$-contracted set is $\frac{i}{i-1}$. Hence, the lower bound follows. (The proof of this lemma also appears in previous related papers such as~\cite{Carmi07,Chen05,Rong04}.)
\end{proof}

Intuitively, the main ingredient of the algorithm is the way we select our contracted sets, which guarantees that each contracted set that is found in $G_R$ saves high transmission range assignments for an optimal solution for $G_R$. Below we formalize this ingredient.

Let $\C =  \{ C_1, C_2, \dots , C_k \}$ be a set of $k$ components in $CG_R$, let $Q$ be a $k$-contracted set of $\C$, let $v_j$ be the node in 
$Q \cap C_j$ for each $C_j \in C$, and let $S$ be a $k$-contractible structure induced by $Q$.

\begin{observation}\label{obs:observation1.3}
Let $\ell$ be the number of leaves in $S$. Then, $$OPT(G_{R\cup Q}) \le OPT(G_{R}) - \ell.$$
\end{observation}
\begin{corollary}\label{cor:corollary4.1}
Let $\mathcal{L} _i$ denote the number of leaves contracted in the $k-i$ iteration. Then, $$OPT(G_{R_{i-1}}) \le OPT(G_{R_i}) - \mathcal{L} _i.$$
\end{corollary}

\begin{observation}\label{obs:observation1.1}
In $G_{R\cup Q}$, if there exists a node $v$ in an optimal solution for $G_{R}$ that induces only edges of the clique over $\C$ (i.e., $v$ reaches only components of $\C$), then $OPT(G_{R\cup Q}) < OPT(G_R)$.
\end{observation}

\begin{observation}\label{obs:observation1.2}
Let $v_i\in C_i\cap Q$ be a node that reaches only one component $C_j\in \C$; then (i) any path from $C_i$ to $C_j$ via $C'\notin \C$ in $CG_R$ must contain $C_k\in \C$, and (ii) for each node $u\in C_i$, by assigning a high transmission range to $u$, if $u$ reaches $C_j$ then every component $C'\notin \C$ is not reachable from $u$. 
\end{observation}
For simplicity of presentation we prove Lemma~\ref{lem:lemma1.4} and Lemma~\ref{lem:lemma1.5} for $k=8$, therefore, the approximation ratio we obtained is based on $k=8$.
However, even-though we prove the lemmas for $k=8$,  the lemmas hold for greater values of $k$, therefore, we keep the statements of the lemmas in a general formulation.


Let $Q_i$ denote an $i$-contracted set that is found during the $k-i$ iteration, and let $S_i$ denote the contractible structure induced by $Q_i$. Recall that $b_i$ (resp., $b_{i,j}$) denote the number of $i$-contracted sets (resp., $j$-contracted sets) found by the algorithm in the $k-i$ iteration. 
\begin{lemma} \label{lem:lemma1.4}
For each $4 \leq i \leq \left\lceil {k}/{2}\right\rceil$, we have 
$$OPT(G_{R_{i-1}}) \le OPT(G_{R_i}) - 2b_i - 2\sum\limits_{j=4}^{i-1}{b_{i,j}}.$$
\end{lemma}
\begin {proof}
Recall that we put $k=8$. Thus, $i=4$ and $\sum\limits_{j=4}^{i-1}{b_{i,j}} =0$. Let $G_{R^{'}_4}$ be the graph in which a contractible structure $S_4$ is found. 
We need to show that $S_4$ saves two to $OPT$ of the remain graph, that is $OPT(G_{R^{'}_4}) \geq  OPT(G_{R^{'}_4\cup Q_4}) +2$. If $S_4$ has two leaves, then, by Observation~\ref{obs:observation1.3}, $S_4$ saves two to $OPTG_{R^{'}_4}$. Therefore, $S_4$ has at most one leaf and there are two such contractible structures, and, since there is no contractible structures of size greater than $4$ in $G_{R_i}$ (and in particular in $G_{R^{'}_4}$), $S_i$ saves two to $OPT(G_{R^{'}_4})$. 
\end{proof}
\begin{lemma} \label{lem:lemma1.5}
For each $\left\lceil {k}/{2}\right\rceil < i < k$, we have 
$$OPT(G_{R_{i-1}}) \le OPT(G_{R_i}) - b_i - 2\sum\limits_{j=4}^{\left\lceil {k}/{2}\right\rceil}{b_{i,j}} - \sum\limits_{j=\left\lceil {k}/{2}\right\rceil + 1}^{i-1}{b_{i,j}}.$$ 
\end{lemma}
\begin {proof}
By Observation~\ref{obs:observation1.3}, we are left with providing a proof for contractible structures $S_i$ without leaves, where $\left\lceil {k}/{2}\right\rceil< i < k$. Let $G_{R^{'}_i}$ be the graph in which $S_i$ is found. 
First, we consider the case where $S_i$ is a simple cycle ($5 \le i \le 7$), and assume towards a contradiction that $OPT(G_{R^{'}_i})=OPT(G_{R^{'}_i\cup Q_i})$. 

Let $H=(\C,E_H )$ be the undirected version of $S_i$ in $CG_{R^{'}_i}$, and let $R^{'*}_i$ be an optimal solution for $G_{R^{'}_i}$. 
Let $G$ be a spanning subgraph of $CG_{R^{'}_i}$, where there is an edge in $G$ between $C_l \in CG_{R^{'}_i}$ and $C_j \in CG_{R^{'}_i}$ if there exists a node  
$v_l \in R^{'*}_i \cap C_l$ that can reach a node in $C_j$ via high transmission range. 

If there exist $C_l \in \C$ and $v \in R^{'*}_i\cap C_l$, such that $v$ can reach only components in $\C$, then by Observation~\ref{obs:observation1.1} we are done. 
Otherwise, let $G^{'}=(\C \cup U',E')$ be a minimum subgraph of $G$ in which all the components in $\C$ are connected, where $U'$ and $E'$ are sets of components and edges in $CG_{R^{'}_i}$, respectively.
Let $U$ be an empty set of nodes. 
For each edge $\{C_l,C_j\}$ of $G^{'}$, such that $C_l,C_j \in \C$, we add a new node $u_{l,j}$ to $U$ and update the set 
$E^{'}$ to be $E^{'}\cup\{ \{C_l,u_{l,j}\},\{u_{l,j},C_j\}\} \setminus \{C_l,C_j\}$.
The obtained graph $G^{'}=(\C \cup U' \cup U,E')$  is a connected graph where $\C$ is an independent set. 
By Lemma~\ref{lem:lemma2.5} and Lemma~\ref{lem:lemma2.4}, the graph $(\C \cup U \cup U',E'\cup E_H)$ contains a \textsc{Sec}-$\C$-\textsc{HamCycle} $H'$. If $H'$ contains nodes from $U'$, then $H'$ admits   
a contracted structure of size at least $i+1$ in $G_{R^{'}_i}$, in contradiction.
Otherwise, $H'$ contains only the nodes of $\C$ and nodes from $U$. Let $H_{\C}$ be the cycle obtained from $H'$ by replacing each pair of consecutive edges $\{C_l,u_{l,j}\},\{u_{l,j},C_j\}$, where $u_{l,j} \in U$ and $C_l, C_j \in \C$, by the edge $\{C_l, C_j\}$. Recall that $\{C_l, C_j\}$ is an edge in $G_{R^{'}_i}$ and for each $C_l \in \C$,  there exists a node  $v \in R^{'*}_i\cap C_l$,  such that $v$ can reach a component $C' \notin \C$ (i.e., $C' \in U'$). Thus, the cycle $H_{\C}$ with the component $C'$ is a contracted structure of size at least $i+1$ in $G_{R^{'}_i}$, in contradiction. 

We now consider contractible structure $S_i$ ($i < 8$) that is neither a simple cycle nor a structure with leaves, that is $S_i$ is a contractible structure of combined (overlapping) simple cycles $\{\C_1, \C_2, \dots, \C_m \}$. 
W.l.o.g., let $\C_1 \subset S_i$ be a simple cycle such that $\{\C_2, \dots, \C_m \}$  is a  contractible structure.
Let $\C= \C_1$ and $\C' = \{\C_2, \dots, \C_m \}$.
Moreover, let $(C_1, \dots, C_t )$  be the components in $\C \setminus  \C' $, such that there exists a component in $\C'$ that has a directed edge to $C_1$ and there is a directed edge from $C_t$ to a component in $\C'$,
 see Figure~\ref{fig:smallestCycleinCombin} for illustration. 

\begin{figure}[h]
   \centering
       \includegraphics[width=0.55\textwidth]{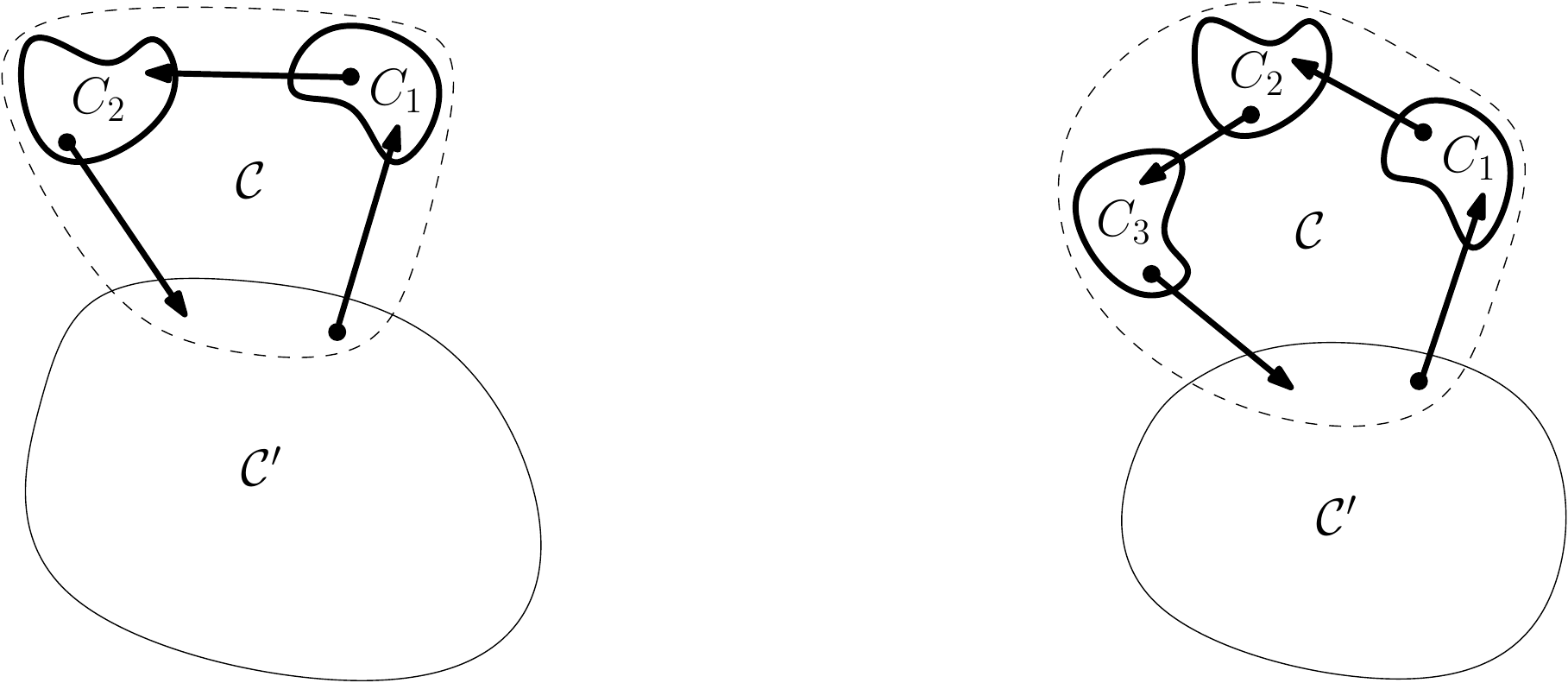}
   \caption{Ilustration of two  contractible structures, where on the left a contractible structure with a leaf, and on the right a contractible structure of combined cycles. }
   \label{fig:smallestCycleinCombin}
\end{figure}

Notice that, if $1 \leq t \leq 2$, then $C_t$ is a leaf, thus $S_i$ is a contractible structure with a leaf. 
However, contractible structures with a leaf have already been considered, 
therefore $t \ge 3$.  Moreover, since $i < 8$ and $t \geq 3$, $\C'$ is also a simple cycle, 
thus the number of components in $(\C' \setminus \C)$ is at least 3 (otherwise, $S_i$ is a contractible structure with a leaf). Therefore, 
the number of components in $(\C \setminus \C')$ and  $(\C' \setminus \C)$ is exactly  $3$ (i.e. $t=3)$.

For $C_i \in C$,  let $\delta(C_i,\C')$  be the path that connects $C_i$ to $\C'$ in the optimal solution. 
Then, $\delta(C_3,\C')$ is the path that connects $C_3$ to $\C'$ in the optimal solution. 
Consider the three cases of  $\delta(C_3,\C')$.


\begin{itemize}
 \item  $\delta(C_3,\C')$ does not pass through   $C_1$ nor trough  $C_2$,  then it must go directly to a component in $\C'$ (otherwise, we have a contractible structure of size greater than $i$), thus, by Observation~\ref{obs:observation1.1},  it saves one to $OPT(G_{R_i})$.

	\item  $\delta(C_3,\C')$ passes through $C_2$.
	$C_3$ can not go though another component $C_x \notin S_i$ to $C_2$ since in this case, by replacing the edge $(C_2, C_3)$ with the reverse path from  $(C_3, C_2)$ that goes though $C_x$,
	 we obtain a contractible structure of size greater than $i$. 
	Moreover, $C_3$ can not go to $C_2$ and to another component $C_x \notin S_i$, since the reverse order of $\C$ with $C_x$ admits a contractible structure with a leaf of size at least $5$. 
	However, contractible structures with a leaf have already been considered.
	Thus, by Observation~\ref{obs:observation1.1}, we save one to $OPT(G_{R_i})$.  

	\item  $\delta(C_3,\C')$ passes through $C_1$. Thus, there is a path $\delta_{C_3,C_1}$ from $C_1$ to $C_3$ that does not include components of $S_i \setminus \{C_1, C_3 \}$. Denote by  
	$ \delta_{\overleftarrow{C_3,C_1} }$ 	the reverse path of $\delta_{C_3,C_1}$. 
	Consider  $\delta(C_2,\C')$, following the same ideas of the two aforementioned cases, $\delta(C_2,\C')$ can not pass though neither $C_1$, $C_3$, nor directly to a component in $\C'$.
	Thus, $\delta(C_2,\C')$ goes through another component $C_x \notin S_i$. Therefore, by replacing the path $(C_1,C_2, C_3, C)$ with the path $ \delta_{\overleftarrow{C_3,C_1} }$, the edge $(C_3,C_2)$, and  $\delta(C_2,\C')$,  where $C \in \C'$ is a reachable component from $C_3$ in $S_i$, we obtain a contractible structure of size greater than $i$.
\end{itemize}



\end{proof} 

\begin{theorem}\label{theo:theorem1approxRatio}
The aforementioned range assignment algorithm is an ${11}/{7}$-approximation algorithm for the dual power assignment problem.
\end{theorem}
\begin{proof}
Set $k$ to be $8$ and let $n$ be the number of components of $CG_{\emptyset}$. By Lemma~\ref{lem:lemma1.2}, $\frac{k}{k-1}$ is the amortized cost of each contracted component of a $k$-contracted set. Then, according to the algorithm description,
\begin{align*}
|R| & \le \frac{k}{k-1}(n-n_{k-1})+\sum\limits_{i=4}^{k-1} i \cdot b_i + \sum\limits_{i=5}^{k-1} {\sum\limits_{j=4}^{i-1}{j \cdot b_{i,j} }}+ OPT(G_{R_3}) \,.
\end{align*} 
By Lemma~\ref{lem:lemma1.1}, $n_{k-1}= n_3 + \sum\limits_{i=4}^{k-1} (i-1)\cdot b_i + \sum\limits_{i=5}^{k-1} {\sum\limits_{j=4}^{i-1}{(j-1) \cdot b_{i,j}}}$, then
\begin{align*}
|R|  \le \frac{k}{k-1}&\left(n- n_3 - \sum\limits_{i=4}^{k-1} (i-1)\cdot b_i - \sum\limits_{i=5}^{k-1} {\sum\limits_{j=4}^{i-1}{(j-1) \cdot b_{i,j}}}\right) \,\\
&+\sum\limits_{i=4}^{k-1} i \cdot b_i + \sum\limits_{i=5}^{k-1} {\sum\limits_{j=4}^{i-1}{j\cdot b_{i,j}}} + OPT(G_{R_3}) \,\\
= \frac{1}{k-1}&\left(k\cdot n-k\cdot n_3 + \sum\limits_{i=4}^{k-1} (k-i)\cdot b_i
+ \sum\limits_{i=5}^{k-1} {\sum\limits_{j=4}^{i-1}{(k-j)\cdot b_{i,j}}}+ (k-1)\cdot OPT(G_{R_3})\right)\, ,
\end{align*}
and, by Lemma~\ref{lem:lemma1.2}, $OPT(G_{R_3}) \le 2(n_3 - 1)$, then  $ n_3 \geq \frac{OPT(G_{R_3})}{2} $, and we have,
\begin{align*}
|R| &\leq \frac{1}{k-1}\left(k\cdot n + \sum\limits_{i=4}^{k-1} (k-i)\cdot b_i
+ \sum\limits_{i=5}^{k-1} {\sum\limits_{j=4}^{i-1}{(k-j)\cdot b_{i,j}}}+  (k-1 - \frac{k}{2}) \cdot OPT(G_{R_3}) \right)\, \\
&= \frac{1}{k-1} \cdot \left(k\cdot n + \sum\limits_{i=4}^{k-1} (k-i)\cdot b_i + \sum\limits_{i=5}^{k-1} {\sum\limits_{j=4}^{i-1}{(k-j)\cdot b_{i,j}}} \right) + \frac{k-2}{2(k-1)} \cdot OPT(G_{R_3}) \, ,
\end{align*}
and, by Lemma~\ref{lem:lemma1.4} and Lemma~\ref{lem:lemma1.5},\\
\begin{eqnarray*}
OPT(G_{R_3}) \le & OPT(G_{R_{k-1}}) - 2\sum\limits_{i=4}^{\left\lceil {k}/{2}\right\rceil} b_i -2\sum\limits_{i=5}^{\left\lceil {k}/{2}\right\rceil}\sum\limits_{j=4}^{i-1}{b_{i,j}} - \sum\limits_{i=\left\lceil {k}/{2}\right\rceil+1}^{k-1} b_i \\
& - 2\sum\limits_{i=\left\lceil {k}/{2}\right\rceil+1}^{k-1}\sum\limits_{j=4}^{\left\lceil {k}/{2}\right\rceil}{b_{i,j}} - \sum\limits_{i=\left\lceil {k}/{2}\right\rceil+1}^{k-1}\sum\limits_{j=\left\lceil {k}/{2}\right\rceil + 1}^{i-1}{b_{i,j}},
\end{eqnarray*}
then,
\begin{align*}
|R| \le & \frac{1}{k-1} \cdot \left(k\cdot n + \sum\limits_{i=4}^{k-1} (k-i)\cdot b_i + \sum\limits_{i=5}^{k-1} {\sum\limits_{j=4}^{i-1}{(k-j) \cdot b_{i,j}}}\right)\\
& \  \qquad + \frac{k-2}{2(k-1)} \cdot \left( OPT(G_{R_{k-1}})- 2\sum\limits_{i=4}^{\left\lceil {k}/{2}\right\rceil} b_i  -2\sum\limits_{i=5}^{\left\lceil {k}/{2}\right\rceil}\sum_{j=4}^{i-1}{b_{i,j}} \right.\nonumber\\ 
& \qquad \left. {} - \sum\limits_{i=\left\lceil {k}/{2}\right\rceil+1}^{k-1} b_i - 2\sum\limits_{i=\left\lceil {k}/{2}\right\rceil+1}^{k-1}\sum\limits_{j=4}^{\left\lceil {k}/{2}\right\rceil}{b_{i,j}} - \sum\limits_{i=\left\lceil {k}/{2}\right\rceil+1}^{k-1}\sum\limits_{j=\left\lceil {k}/{2}\right\rceil + 1}^{i-1}{b_{i,j}} \right).
\end{align*}
Since 
$$\frac{1}{k-1} \cdot \sum\limits_{i=4}^{k-1} (k-i)\cdot b_i - 
\frac{k-2}{2(k-1)} \cdot \left( 2\sum\limits_{i=4}^{\left\lceil {k}/{2}\right\rceil} b_i  
                              + \sum\limits_{i=\left\lceil {k}/{2}\right\rceil+1}^{k-1} b_i \right) \le 0 $$
and 
%
\begin{align*}
\frac{1}{k-1} \cdot \left(  \sum\limits_{i=5}^{k-1} {\sum\limits_{j=4}^{i-1}{(k-j) \cdot b_{i,j}}} \right) & - 
    \ \frac{k-2}{2(k-1)} \cdot \left(  2\sum\limits_{i=5}^{\left\lceil {k}/{2}\right\rceil}\sum_{j=4}^{i-1}{b_{i,j}} \right.\nonumber\\  &  \left. {} 
  +   2\sum\limits_{i=\left\lceil {k}/{2}\right\rceil+1}^{k-1}\sum\limits_{j=4}^{\left\lceil {k}/{2}\right\rceil}{b_{i,j}} + \sum\limits_{i=\left\lceil {k}/{2}\right\rceil+1}^{k-1}\sum\limits_{j=\left\lceil {k}/{2}\right\rceil + 1}^{i-1}{b_{i,j}} \right)  \le 0, 
\end{align*}
we have    
\old{
\begin{align*}
= & \frac{1}{k-1}\cdot \left( k\cdot n + \frac{k-2}{2}OPT(G_{R_{k-1}}) + \sum\limits_{i=4}^{k-1} (k-i)\cdot b_i - \sum\limits_{i=4}^{\left\lceil {k}/{2}\right\rceil} (k-2)\cdot b_i\, \right.\nonumber\\ 
& \qquad \left. {} - \frac{1}{2} \sum\limits_{i=\left\lceil {k}/{2}\right\rceil+1}^{k-1} (k-2)\cdot b_i + \sum\limits_{i=5}^{k-1} {\sum\limits_{j=4}^{i-1}{(k-j) \cdot b_{i,j}}} - \sum\limits_{i=5}^{\left\lceil {k}/{2}\right\rceil}\sum\limits_{j=4}^{i-1}{(k-2) b_{i,j}} \,\right.\nonumber\\ 
& \qquad \left. {}
- \sum\limits_{i=\left\lceil {k}/{2}\right\rceil+1}^{k-1}\sum\limits_{j=4}^{\left\lceil {k}/{2}\right\rceil}{(k-2)b_{i,j}} -\frac{1}{2}\sum\limits_{i=\left\lceil {k}/{2}\right\rceil+1}^{k-1}\sum\limits_{j=\left\lceil {k}/{2}\right\rceil + 1}^{i-1}{(k-2) b_{i,j}} \right)\\
 = & \frac{1}{k-1}\cdot \left( k\cdot n + \frac{k-2}{2}OPT(G_{R_{k-1}}) + \sum\limits_{i=4}^{\left\lceil {k}/{2}\right\rceil} (2-i)\cdot b_i + \sum\limits_{i=\left\lceil {k}/{2}\right\rceil+1}^{k-1} (\frac{k}{2}+1-i)\cdot b_i\, \right.\nonumber\\
 &\qquad \left. {} +\sum\limits_{i=5}^{\left\lceil {k}/{2}\right\rceil}\sum\limits_{j=4}^{i-1}{(2-j) b_{i,j}} + \sum\limits_{i=\left\lceil {k}/{2}\right\rceil+1}^{k-1}\sum\limits_{j=4}^{\left\lceil {k}/{2}\right\rceil}{(2-j)b_{i,j}} 
\,\right.\nonumber\\ 
& \qquad \left. {}
+ \sum\limits_{i=\left\lceil {k}/{2}\right\rceil+1}^{k-1}\sum\limits_{j=\left\lceil {k}/{2}\right\rceil+1}^{i-1}{(\frac{k}{2}+1-j)b_{i,j}}\right).
\end{align*}
Notice that all $b_i$ and $b_{i,j}$'s coefficients are negative, thus we have
}
$$|R| \le \frac{k}{k-1} \cdot n + \frac{k-2}{2(k-1)} \cdot OPT(G_{R_{k-1}}).$$

Finally, since $OPT(G_{R_{\emptyset}}) \ge n$ and $OPT(G_{R_{\emptyset}}) \ge OPT(G_{R_{k-1}})$, we have
$$|R| \le \frac{3k-2}{2(k-1)} \cdot OPT(G_{R_{\emptyset}}), $$  
thus, for $k=8$, we have $$ |R| \le \frac{11}{7}\cdot OPT.$$
\end{proof}

\subsection{Optimal solution for $G_{R_3}$}\label{sec:OptG3}

Given a set $P$ of $n$ wireless nodes in the plane, two transmission ranges $r_L$ and $r_H$, and $R_3 \subseteq P$, such that  $G_{R_3}$ does not contain a contracted set of size greater than $3$.  Then finding a minimum set $R^*_3\subseteq P$, such that the induced communication graph $G_{R_3\cup R^*_3}$ is strongly connected can be done in polynomial time.

Our algorithm is based on the idea of Carmi and Katz~\cite{Carmi07} and works as follows. Set $R=\emptyset$ and compute the induced communication graph $G_{R_3}$ by assigning $r_H$ to each node in $R_3$ and assigning $r_L$ to each node in $P\setminus R_3$. Next, while $G_{R_3\cup R}$ contains a $3$-contracted set forming a simple cycle, find such a contracted set, and add its $3$ nodes to $R$. When $G_{R_3\cup R}$ does not contain a $3$-contracted set forming a simple cycle, it induces a tree of well-separated $j$-contracted sets, we solve the subproblem in each strongly connected component of $G_{R_3\cup R}$ independently, and add to $R$ the nodes that are in the solution.

Notice that the resulting $CG_{R_3\cup R}$ has one component, and, therefore, $G_{R_3\cup R}$ is strongly connected. In the following, we prove that this algorithm solves the problem optimally, i.e., $|R|=|R^*_3|$.

Let $\C=\{C_1,C_2,C_3\}$ be a set of $3$ components in $CG_{R_3\cup R}$ and let $Q$ be a $3$-contracted set of $\C$, such that the $3$-contractible structure induced by $Q$ forms a simple cycle. The following two observations follow from the fact that the graph $CG_{R_3\cup R}$ does not contain a contracted set of size greater than $3$.

\begin{observation}\label{lem:lemma1.6}
By adding the nodes in $Q$ to $R$ the problem is separated into at least three independent subproblems. I.e., by removing the components in $\C$ and the edges incident to them from $CG_{R_3\cup R}$, the graph $CG_{R_3\cup R}$ remains with at least three connected components.
\end{observation}

\begin{observation}\label{obs:observation1.7}
There exists an optimal solution $R^*_3$ for $G_{R_3}$ that contains the nodes in $Q$.
\end{observation}

When $G_{R_3\cup R}$ does not contain a $3$-contracted set forming a simple cycle, it induces a tree of well-separated $j$-contracted sets. Thus, assigning a high transmission range to a node in one strongly connected component cannot result in forcing an assignment of a high transmission range to a node in another strongly connected component. Therefore, each strongly connected component of $G_{R_3\cup R}$ is an independent subproblem. Each node in a strongly connected component of $G_{R_3\cup R}$ can reach at most two other strongly connected components via high transmission range. Hence, each strongly connected component is an instance of the \emph{2 set cover} problem, which can be solved optimally.

Thus, we conclude that the algorithm described above solves the problem optimally.


\section{Application of a second Hamiltonian cycle to SCSS}\label{sec:conjecture}
In this section, we show that the correctness of Conjecture~\ref{con:conjecture2.1} implies that the approximation algorithm of Khuller et al.~\cite{Young02}, which achieves a performance guarantee of  $\approx 1.61$ for the SCSS problem, is a $3/2$-approximation algorithm in symmetric unweighted digraphs. This matches the best known approximation ratio for this problem, achieved by Vetta~\cite{Vetta01}. Even though Vetta's result is very novel, it is much more complicated. 

Given a strongly connected graph, the algorithm finds a cycle of length at least some constant $k$ while there exists such a cycle, and then a longest cycle in the current graph, contracts the cycle, and recurses. The contracted graph remains strongly connected. When the graph, finally, collapses into a digraph with cycles of length at most $3$, it solves the subproblem optimally and returns the set of edges contracted during the course of the algorithm as the desired SCSS.

This algorithm differs from the DPA algorithm (described in Section~\ref{sec:dpa}) in the contracted structures. More precisely, only simple cycle structures are found (since simple cycle structures are the only contracted structures exist). Thus, assuming Conjecture~\ref{con:conjecture2.1} holds, each structure found during the algorithm saves at least two edges for an optimal solution. This implies the following lemma that is similar but stronger than Lemma~\ref{lem:lemma1.4}.

\begin{lemma} \label{lem:SCSSlemma}
For each $4 \leq i \leq k-1$, we have \ $OPT(G_{i-1}) \le OPT(G_i) - 2b_i $, where $OPT(G_i)$ is the size of an optimal solution for the component graph at the beginning of the $k-i$ iteration, and $b_i$ is the number of contracted structures found and contracted by the algorithm in the $k-i$ iteration.
\end{lemma}

By combining this lemma with Lemma~\ref{lem:lemma1.1} and Lemma~\ref{lem:lemma1.2}, we get the following theorem.
\begin{theorem}\label{thm:theorem1approxRatio}
The algorithm of Khuller et al. in~\cite{Young02} (described above) is a ${3}/{2}$-approximation algorithm for the SCSS problem in symmetric unweighted digraphs, assuming Conjecture~\ref{con:conjecture2.1} holds.
\end{theorem}
\begin{proof}
Applying a similar (yet simpler) analysis of the performance of the dual power assignment algorithm (Section~\ref{sec:dpa}) yields an upper bound of $\frac{3k-2}{2(k-1)}OPT$. This approximation ratio tends to $3/2$ as $k$ increases. 
\end{proof}

Since we verified Conjecture~\ref{con:conjecture2.1} for $|V| < 24$ (see Lemma~\ref{lem:lemma24}), 
we have the following corollary.
\begin{corollary}
The algorithm of Khuller et al. in~\cite{Young02} is a ${35}/{23}$-approximation algorithm ($\thickapprox 1.522$) for the SCSS problem in symmetric unweighted digraphs.
\end{corollary}

\subsection{SCSS for symmetric digraphs with bounded cycle length}
In~\cite{Khuller02}, Khuller et al. consider the SCSS problem in a strongly connected digraphs with bounded cycle length. They give a proof that, for graphs where each directed cycle has at most three edges is equivalent to the maximum bipartite matching, and, thus can be solved optimally.
Moreover, in \cite{Young02} Khuller et al. prove that the problem remains NP-hard even when the maximum cycle length is at most five.
In this section, we consider the same problem in symmetric digraphs with bounded cycle length, and show the following.
\begin{theorem}
The algorithm of Khuller et al. in~\cite{Young02} (described above) is a $\frac{3k -2}{2k}$-approximation algorithm for the SCSS problem in symmetric unweighted digraphs, where $k$ is the maximum cycle length in the graph, assuming Conjecture~\ref{con:conjecture2.1} holds or $k < 24$.
\end{theorem}
\begin{proof}
The length of the longest cycle is at most $k$, thus the first phase of the algorithm (looking for cycles of length greater than $k$) is redundant.
Therefore, we have 
\begin{align*}
|R| & \le \sum\limits_{i=4}^{k} i \cdot b_i + OPT(G_{R_3}) \, \\
    & = \sum\limits_{i=4}^{k} i \cdot b_i + \frac{OPT(G_{R_3})}{2} + \frac{OPT(G_{R_3})}{2} \, \\
    & \stackrel{(1)}{\le} \sum\limits_{i=4}^{k} i \cdot b_i + \frac{OPT(G_R) -2\sum\limits_{i=4}^{k} \cdot b_i }{2} + \frac{OPT(G_{R_3})}{2} \, \\
    &  \stackrel{(2)}{\le} \sum\limits_{i=4}^{k} i \cdot b_i + \frac{OPT(G_R) -2\sum\limits_{i=4}^{k} \cdot b_i }{2} + \frac{2(n_3 -1)}{2} \, \\
    &  \stackrel{(3)}{=}  n - n_3 + \sum\limits_{i=4}^{k}  \cdot b_i + \frac{OPT(G_R)}{2} - \sum\limits_{i=4}^{k} \cdot b_i + n_3 - 1 \, \\
    & = n +      \frac{OPT(G_R)}{2} - 1 \, \\
    &  \stackrel{(2)}{\le} OPT(G_R) \cdot ( \frac{k-1}{k} + \frac{1}{2}) \, \\
    & = \frac{3k-2}{2k} \cdot OPT(G_R)
\end{align*}  
where (1) follows from Lemma~\ref{lem:SCSSlemma}, (2) follows from Lemma~\ref{lem:lemma1.2}, and (3) follows from Lemma~\ref{lem:lemma1.1}.
 
\end{proof}

\section*{Acknowledgment}
The authors would like to thank Carsten Thomassen for his help with the correctness of Lemma~\ref{lem:lemma2.1}.

\bibliographystyle{plain}
\bibliography{ref}
\end{document}